\documentclass{wscpaperproc}
\usepackage{latexsym}
\usepackage{graphicx}
\usepackage{mathptmx}

\usepackage{amsmath}
\usepackage{amsfonts}
\usepackage{amssymb}
\usepackage{amsbsy}
\usepackage{amsthm}
\usepackage{color}
\usepackage{empheq} 
\usepackage{epsfig} % for postscript graphics files
\usepackage{graphicx} % for pdf, bitmapped graphics files
\usepackage{subcaption}
\usepackage{times} % assumes new font selection scheme installed
\usepackage{wrapfig}% embedding figures/tables in text (i.e., Galileo style)
\usepackage{multicol}
\usepackage{bbm}
\usepackage{caption}
\usepackage{cite}
\usepackage{relsize}
\usepackage{array}
\usepackage[]{algorithm2e}
\usepackage{algorithmic}

\usepackage[pdftex,colorlinks=true,urlcolor=blue,citecolor=black,anchorcolor=black,linkcolor=black]{hyperref}

\newtheoremstyle{wsc}% hnamei
{3pt}% hSpace abovei
{3pt}% hSpace belowi
{}% hBody fonti
{}% hIndent amounti1
{\bf}% hTheorem head fontbf
{}% hPunctuation after theorem headi
{.5em}% hSpace after theorem headi2
{}% hTheorem head spec (can be left empty, meaning `normal')i

\theoremstyle{wsc}
\newtheorem{theorem}{Theorem}

\newtheorem{lemma}{Lemma}

\renewcommand{\vec}[1]{\mathbf{#1}}
\newcolumntype{M}{>{\centering\arraybackslash}m{\dimexpr.25\linewidth-2\tabcolsep}}

\providecommand{\shortcite}[1]{\shortcite{#1}}

\usepackage[capitalise]{cleveref}

\graphicspath{{Media/}}

\begin{document}

%***************************************************************************
% AUTHOR: AUTHOR NAMES GO HERE
% FORMAT AUTHORS NAMES Like: Author1, Author2 and Author3 (last names)
%
%		You need to change the author listing below!
%               Please list ALL authors using last name only, separate by a comma except
%               for the last author, separate with "and"
%
\WSCpagesetup{Alaeddini and Klein}

% AUTHOR: Enter the title, all letters in upper case
%\title{PSPO: Stochastic Optimization Applied to an Epidemiological Model}
\title{APPLICATION OF A SECOND-ORDER STOCHASTIC OPTIMIZATION ALGORITHM FOR FITTING STOCHASTIC EPIDEMIOLOGICAL MODELS}

% AUTHOR: Enter the authors of the article, see end of the example document for further examples
\author{Atiye Alaeddini\\ 
Daniel J. Klein \\[12pt]
%\vspace{12pt}
Institute for Disease Modeling \\
3150 139th Ave SE, Building 4\\
Bellevue, WA, 98005, USA\\
}

\maketitle

\section*{ABSTRACT}
Epidemiological models have tremendous potential to forecast disease burden and quantify the impact of interventions.  Detailed models are increasingly popular, however these models tend to be stochastic and very costly to evaluate.  Fortunately, readily available high-performance cloud computing now means that these models can be evaluated many times in parallel. Here, we briefly describe PSPO, an extension to Spall's second-order stochastic optimization algorithm, Simultaneous Perturbation Stochastic Approximation (SPSA), that takes full advantage of parallel computing environments. The main focus of this work is on the use of PSPO to maximize the pseudo-likelihood of a stochastic epidemiological model to data from a 1861 measles outbreak in Hagelloch, Germany. Results indicate that PSPO far outperforms gradient ascent and SPSA on this challenging likelihood maximization problem.

\section{INTRODUCTION}
\label{sec:intro}

In public health, it is critical to have a reasonable understanding of an epidemic disease in order to set pragmatic goals and design highly-impactful and cost-effective interventions. Mathematical models of these epidemiological processes can support decision making by forecasting disease spread in space and time, and by evaluating intervention outcomes many times in-silico before spending valuable resources implementing real-world programs. Recent efforts in computational epidemiology have focused on the design and application of detailed stochastic models that capture physical mechanisms through which disease propagates, along with the statistical fluctuations inherent in complex systems. These stochastic models are readily available, and recent work has focused on applying these models to malaria \shortcite{eckhoff2016impact,eckhoff2013mathematical,marshall2016key,gerardin2016optimal}, HIV \shortcite{bershteyn2016age,eaton2015assessment}, polio \shortcite{mccarthy2016spatial,grassly2006new}, and more.

A central challenge in working with these complex models lies in finding input parameters that result in model outputs closely matching observed data. The stochastic nature of these models means that each set of input parameters maps to a distribution of outcomes, from which each sample (model run) can take several hours to obtain. In the quick-to-evaluate deterministic model setting, almost any classical optimization algorithm, such as steepest descent or Newton-Raphson, can be used to fit the model to data by maximizing the pseudo-likelihood of the parameters given the data. However, care must be taken when applying deterministic methods to stochastic objective functions, as the inherent noise causes unexpected behavior. Some have tried model averaging, wherein the results at a given point are averaged over many replicates to reduce stochastic fluctuations. While this approach mitigates stochasticity, it is expensive and does not resolve the problem. Further, algorithms that evaluate the model one point at a time, like Markov chain Monte Carlo and most optimization algorithms, will simply take too long to converge when applied to these time- and memory-expensive models.

There exist many optimization methodologies for both deterministic and stochastic systems. The steepest descent method \shortcite{Nocedal99} is the most prominent iterative method for optimizing a complex objective function. The gradient-based algorithms, such as Robbins-Monro \shortcite{robbins1951stochastic}, Newton-Raphson \shortcite{froberg1969introduction}, and neural network back-propagation \shortcite{rumelhart1985learning}, rely on direct measurements of the gradient of the objective function with respect to the optimization parameter. But, in many cases the gradient of the loss function is not available. This is a common occurrence, for example, in complex systems, such as the optimization problem given in \shortcite{tsilifis2017efficient}, the exact functional relationship between the loss function value and the parameters is not known, and the loss function is evaluated by measurements on the system or by running simulation. A review on the main areas of optimization via simulation can be found in \shortcite{fu1994optimization,hong2009brief,swisher2000survey}.

Some algorithms have been developed specifically to optimize stochastic black-box cost functions. Of these, the Kiefer-Wolfowitz algorithm \shortcite{kiefer1952stochastic} is perhaps the most well known. This algorithm estimates the gradient from noisy measurements using finite differencing. Another well known stochastic optimization algorithm in the case of high dimensional problems is Simultaneous Perturbation Stochastic Approximation (SPSA), which is an approximation algorithm based on simultaneous perturbation \shortcite{spall1992multivariate}. Later, J.~C.~Spall presented a second-order variant of SPSA \shortcite{spall2000adaptive}. The SPSA algorithm is used extensively in many different areas, e.g.\, signal timing for traffic control \shortcite{ma2013solving}, and some large scale machine learning problems \shortcite{byrd2011use}. The convergence of this algorithm to the optimal value in the stochastic \emph{almost sure} sense makes it suitable in many applications. Despite these advantageous properties, SPSA is a serial algorithm that evaluates the (stochastic) function a few points at a time. Specialized algorithms, like the PSPO algorithm \shortcite{alaeddini2017pspo} we employ here, are thus needed to take advantage of fundamentally-parallel resources like high-performance cloud-based computing. Other researchers have looked at ways of enhancing the convergence of the SPSA algorithm, e.g.\, iterate averaging is an approach aimed at achieving higher convergence rate in a stochastic setting \shortcite{polyak1992acceleration}. Another variant of SPSA in which the parameter perturbations are based on deterministic, instead of random, perturbations was considered in \shortcite{bhatnagar2003two}. Several other modified versions of SPSA were proposed in the literature \shortcite{kocsis2006universal,spall2009feedback}.

Among epidemiological parameter search algorithms, some have focused on finding {\em all} parameter combinations that are consistent with observed data. These algorithms, such as basic sampling with importance sampling, Bayesian History Matching \shortcite{andrianakis2015bayesian}, and Incremental Mixture Importance Sampling \shortcite{wagner2014quantifying} have been used effectively. However, the challenging task of finding all parameter combinations means that these algorithms are only applicable in lower dimensional parameter spaces, typically $20$ or less. Other approaches have focused on finding parameter combinations that achieve a binary goal-oriented outcome, e.g.~disease eradication, at a specified probability \shortcite{roh2014stochastic,klein2014separatrix}. When the population counts are sufficiently large, mean field theory can be used to give an approximate model, popular among researchers who study disease spread in a network \shortcite{alaeddini2016optimal,preciado2013optimal,miller2012edge}. While all of these algorithms are useful in specific situations, none address the challenge of finding the best inputs, especially for models that are stochastic and have many parameters.

%We formulate both the stochastic Gillespie stochastic simulation algorithm with the binomial tau-leap approximation \shortcite{gillespie2001approximate}. 

%moment closure \shortcite{hasenauer2014method}.

Uncertainty in stochastic modeling is either fluctuation-driven (intrinsic) or parameter-driven (extrinsic). Both sources of uncertainty are important to quantify, as the noise distributions can impact results \shortcite{bayati2012influence}, and intervention impact can be sensitive to parameters. To approximate extrinsic uncertainty in an optimization framework, we make use of the Cramer-Rao Bound \shortcite{rao1992information}. This bound is derived from the inverse of the Fisher Information Matrix, which evaluates the expected curvature of the log-likelihood at the point of maximum value. High curvature is indicative of easily identified parameters, and thus a tight error bound.

The main contribution of this paper is the application of PSPO, a novel second-order parallel-friendly algorithm, to fitting a stochastic model to historical measles data from Hagelloch, Germany. The PSPO algorithm we use here, take advantage of parallel cloud-based computing resources, which is appropriate for the specific application described in this paper. We also use the relationship between the number of parallel rounds of computation and error tolerance of the gradient for each iteration in PSPO algorithm to achieve an acceptable convergence rate for this particular application. We demonstrate that PSPO works well for fitting a stochastic spidemiological model and compares favorably to the conventional simultaneous perturbation optimization algorithm (SPSA).

The rest of the paper is structured as follows. In \cref{sec:prelim}, we give a summary of the simultaneous perturbation optimization method. The model under study is introduced in \cref{sec:problem}. A review of the modified stochastic optimization algorithm, called PSPO, is presented in \cref{sec:SPSApp}. The numerical simulations are given in \cref{sec:sims}, and \cref{sec:conclusion} concludes the paper.

\section{PRELIMINARIES: SIMULTANEOUS PERTURBATION STOCHASTIC APPROXIMATION}
\label{sec:prelim}

%\subsection{Stochastic Simulation Algorithm}

%\subsection{Simultaneous Perturbation Stochastic Approximation}

Consider the problem of maximization of a reward function $L(\theta): \mathbb{R}^p \rightarrow \mathbb{R}$. J. C. Spall proposed an efficient stochastic algorithm called Simultaneous Perturbation Stochastic Approximation (SPSA) \shortcite{spall1992multivariate,spall2000adaptive,spall2009feedback}. The SPSA algorithm efficiently estimates the gradient and the Hessian matrix using finite difference techniques. This algorithm basically consists of two parallel recursions for estimating the optimization parameter, $\theta$, and the Hessian matrix, $H(\theta)$. The first recursion is a stochastic equivalence of the Newton-Raphson algorithm, and the second one estimates the Hessian matrix. The two recursions are
\begin{equation} \label{spsa2}
\begin{aligned}
	&\theta_{k+1} = \theta_{k} + a_k \left( \Pi_{\mathcal{P}}  (\bar{H}_k)  \right)^{-1} \hat{\vec{g}}_k(\theta_k) \,, \\
	&\bar{H}_k =  \frac{k}{k+1} \bar{H}_{k-1} + \frac{1}{k+1} \hat{H}_{k} \,, %\left( \hat{H}_{k} - \Psi_k (\Pi_{\mathcal{P}}  (\bar{H}_k) ) \right)
\end{aligned}
\end{equation}
where $a_k$ is a positive scalar factor, $\mathcal{P}$ denotes the set of all negative definite matrices, and $\Pi_{\mathcal{P}} (\cdot)$ is the projection into the admissible set $\mathcal{P}$. Here, $\hat{\vec{g}}_k$ and $\hat{H}_{k}$ are the estimated gradient and Hessian at iteration $k$. The SPSA approach for estimating the $\hat{\vec{g}}_k(\theta)$ and $\hat{H}_k(\theta)$ follows.

Let $\Delta_k \in \mathbb{R}^p$ be vectors of $p$ mutually independent zero-mean random variables satisfying the condition of $E\{\Delta_k^{-1} \}$ be bounded. An admissible distribution is a Bernoulli $\pm 1$ distribution. The two-sided estimate of the gradient at iteration $k$ is given by:
\begin{equation} \label{est_g_1st}
	\hat{\vec{g}}_k(\theta_k) =  \frac{L(\theta+c_k \Delta_k) - L(\theta-c_k \Delta_k) }{2c_k} \Delta_k^{-1} \,, 
\end{equation}
where $c_k$ is a positive scalar. Note that in \eqref{est_g_1st}, $\Delta_k^{-1}$ is the element-wise inverse of $\Delta_k$. For second order SPSA, it is suggested to using a one-sided gradient approximation, given by:
\begin{equation} \label{est_g_1sided}
	\hat{\vec{g}}_k(\theta_k) =  \frac{L(\theta+c_k \Delta_k) - L(\theta)}{c_k} \Delta_k^{-1} \,.
\end{equation}

Now let $\tilde{\Delta}_k \in \mathbb{R}^p$ be vectors of random variables satisfying the same condition of $\Delta_k$. The positive scalars $\tilde{c}_k$ are chosen to be smaller than $c_k$. The following formula gives a per-iteration estimate of the Hessian matrix.
\begin{equation} \label{est_H}
	\hat{H}_k =  \frac{1}{2} \left[ \frac{\delta G_k}{2\tilde{c}_k} \tilde{\Delta}_k^{-1} + \left( \frac{\delta G_k}{2\tilde{c}_k} \tilde{\Delta}_k^{-1} \right)^T \right] \,, 
\end{equation}
where 
\begin{equation} 
	\delta G_k = \hat{\vec{g}}_k(\theta_k + \tilde{c}_k \tilde{\Delta}_k ) - \hat{\vec{g}}_k(\theta_k - \tilde{c}_k \tilde{\Delta}_k ) \,.
\end{equation}

The practical implementation details of this algorithm are given in \shortcite{spall2000adaptive,spall2012stochastic}.

%%%%%%%%%%%%%%%%%%%%%%%%%%%%%%%%%%%%%%%%%%%%%%%%%%%%%%%%%%%%%
\section{PROBLEM DEFINITION} \label{sec:problem}
%%%%%%%%%%%%%%%%%%%%%%%%%%%%%%%%%%%%%%%%%%%%%%%%%%%%%%%%%%%%%

In this paper, we formulate our problem for virus spread in a population consists of $N$ individuals. In order to have a reliable predictive model, we need to estimate the model parameters using noisy prevalence data. Two main sources of prevalence data are clinics and health surveys. The clinic data are often available for several time periods. But, health survey data are usually available at only one or a few number of time points. We do not make any assumption on whether the prevalence data is sparse or dense. The output, $\vec{\zeta}$, is the prevalence rates at time points, $t_1, t_2, \cdots, t_n$. To estimate the model parameters, we use maximum likelihood estimate (MLE). Assume the likelihood function of a given $\vec{\theta}$ for a given observation $\vec{\zeta}$ is expressed as:
\begin{equation} 
	l(\vec{\theta} \mid \zeta_1, \zeta_2, \cdots, \zeta_n) = \Pr( \zeta_1, \zeta_2, \cdots, \zeta_n \mid \vec{\theta})  \,,
\end{equation}
where, $\zeta_j$ is the observed prevalence at time $t_j$. Then the log likelihood function is given by 
\begin{equation} \label{costLogReg}
	L(\vec{\theta}) =  \log l(\vec{\theta} \mid \zeta_1, \zeta_2, \cdots, \zeta_n) = \sum_{j=1}^n \log \Pr(\zeta_j \mid \vec{\theta})  \,.
\end{equation}

In order to find the probability, we use a beta binomial model. The beta binomial model is a type of Bayesian model. In our model, the probability of having $I$ infected individuals among $N$ people can be expressed as a binomial distribution
$$ I \mid N, \kappa \sim Binomial(\kappa,N) \,$$
where $\kappa$ is the probability of each individual being infected. Thus the binomial likelihood is
$$ \Pr(I\mid \kappa,N)=\frac{N!}{I!(N-I)!}\kappa^I(1-\kappa)^{N-I} \,.$$
Now, $\kappa$ itself is a probability that can vary over the interval $[0, 1]$, and a Beta distribution is used as the prior
$$ \kappa \mid \beta_1, \beta_2 \sim Beta(\beta_1, \beta_2) \,,$$
then, after observing $I$ infected individuals in a simulation, we have
$$ \kappa \mid I, N, \beta_1, \beta_2 \sim Beta(\beta_1+I, \beta_2+N-I) \,.$$

Using the beta binomial distribution, the probability of observing $\zeta_j$ given modeling parameter $\vec{\theta}$ is
\begin{equation} \label{Pr_zj}
\begin{aligned}
	&\Pr(\zeta_j \mid \vec{\theta}) = \int_0^1 \Pr(\zeta_j\mid \kappa,N) \Pr(\kappa\mid I_j,N) d\kappa =\\
	& \frac{N!}{\zeta_j! (N-\zeta_j)!} \frac{\Gamma(2+N) \Gamma(1+I_j+\zeta_j) \Gamma(1+2N-I_j-\zeta_j)} {\Gamma(1+I_j)\Gamma(1+N-I_j) \Gamma(2+2N)}
\end{aligned}	
\end{equation}
where, $I_j$ comes from running our epidemic model on a population of $N$ individuals for time $t\leq t_j$, and substituting $\theta$ as model parameters. In this paper, the evolution of an epidemic disease can be simulated using Gillespie's stochastic simulation algorithm (SSA) \shortcite{gillespie2001approximate}. In fact, we employ a stochastic compartmental model, in which state transitions resemble chemical reactions, although the PSPO methodology works equally well on individual-based models.  
%Forward paths of compartmental models can be simulated using Gillespie's stochastic simulation algorithm (SSA) \shortcite{gillespie2001approximate}. The evolution of the system is described by a set of chemical reactions. Assume a chemical system consisting of a number of molecules and some known reaction channels with given reaction constants \textcolor{red}{propensities}. SSA first samples the time delay until the next reaction using the total propensity, and then selects one reaction to fire, proportional to propensity.  This procedure continues until the number of reactants is zero or the simulation time as been exceeded.

%In the case of noise-free measurement, we can compute the gradient vector $\vec{g}(\theta)$ as
%\begin{equation} \label{gradLogReg}
% \vec{g}(\theta)= \frac{\partial L(\vec{\theta})}{\partial \theta} \,,
%\end{equation}
%and the Hessian matrix as
%\begin{equation} \label{HessLogReg}
% H= \frac{\partial^2 L(\vec{\theta})}{\partial \theta^2}\,.
%\end{equation}
%In the case of deterministic system and noise-free measurement, $H(\theta) \preceq 0, \; \forall \theta$, and
Finally, the optimal model parameters, $\theta^*$, are obtained by solving the following optimization problem:
$$ \theta^* = \arg\max L(\theta) \,,$$
where $L(\theta)$ is given by \eqref{costLogReg}. 
%The availability of the gradient information in the case of noise-free observation allows using gradient descent-based methods \shortcite{wang2010parameter}.
%
%In the case of noisy measurement, noise might alter the optimization process because the algorithm is receiving misleading information which results in computing a wrong search direction, and a totally wrong optimal point. Since our optimization problem, here, is a stochastic system with noisy observation thus the gradient descent type algorithms are not useful. Instead we need an efficient algorithm that considers the uncertainties in the observation. A stochastic optimization algorithm will be presented in the next section.
In the next section we present a robust stochastic algorithm to estimate the optimal parameter, $\theta^*$.

%\subsection{Estimation of Fisher Information Matrix} \label{sec:CramerRao}

In this paper, we are also interested in calculating the confidence region of the estimated optimal solution obtained by our optimization algorithm. The Fisher information matrix is a tool used here to compute the confidence region. Given the information matrix, one can use Cramer-Rao inequality to compute the confidence region of the estimated parameter. Assuming $\theta^*$ is the true value of the parameter, then the uncertainty of the estimated parameter, $\hat{\theta}$, is bounded using the Cramer-Rao inequality as
$$ \Sigma_{\theta}  \geq F^{-1} (\theta^*) \,,$$
where $\Sigma_{\theta}$ is the covariance matrix of $\hat{\theta}$, and $F$ is the Fisher information matrix \shortcite{scharf1991statistical}. In other words, the estimated standard errors of the maximum likelihood estimation are the square roots of the diagonal elements of the inverse of the observed Fisher information matrix. Basically, smaller information matrix means larger variation in the estimated parameter and less confidence in the accuracy of the estimated value.
Having the likelihood function, the $p \times p$ Fisher information matrix is given by:
\begin{equation} \label{fisher_1st}
	F(\theta) \equiv \mathbb{E} \left[ \left. \frac{\partial L(\theta)}{\partial \theta} \cdot \frac{\partial L(\theta)}{\partial \theta^T} \right| \theta \right] \,,
\end{equation}
where $L(\theta) = \log \Pr(\vec{\zeta} \mid \theta) = \log l(\theta \mid \vec{\zeta})$. A well known equivalent form of the fisher information matrix, which is defined based on the assumption of twice differentiability of the log-likelihood function, is:
\begin{equation} \label{fisher}
	F(\theta) \equiv -\mathbb{E} \left[ \left. H(\theta \mid \vec{z}) \right| \theta \right] \,,
\end{equation}
where
$$ H(\theta \mid \vec{\zeta}) \equiv \frac{\partial^2 L}{\partial \theta \partial \theta^T} $$
is the Hessian matrix. The latter form of the Fisher information matrix in \eqref{fisher} is often easier to compute than the former form given in \eqref{fisher_1st}. The analytical calculation of the Fisher information matrix is difficult here because of nonlinearity nature of the the epidemiological model. Thus, we use an approximation method based on a Monte Carlo sampling approach, which is introduced in \shortcite{spall2005monte}. The idea is to estimate the Hessian matrix of the $\log$-likelihood function for a large number of samples and then average the negative of the Hessian matrices to obtain an estimate of the Fisher information matrix. Basically the same approach of simultaneous perturbation used in \eqref{est_H} can be used to estimate the Hessian matrix. A summary of the Monte Carlo sampling approach (from \shortcite{spall2005monte}) is given here.

Let $\vec{z}_k, \ \ k=1,\cdots, K$, be the $k$th Monte Carlo generated sample from running the epidemic disease model using the estimated parameter, $\theta$. To obtain $\vec{z}_k$, we need to run the simulation for the given $\theta$. The pseudo measurement, $\vec{z}_k$, is a realization of the stochastic simulation based on the probability distributions for each node. The Monte Carlo based estimate of the information matrix is given by \shortcite{spall2005monte}:
\begin{equation} \label{ModifiedFisher}
\begin{aligned}
	\bar{F}_{J,k} (\theta) = \frac{k-1}{k} \bar{F}_{J,k-1} (\theta) - \frac{1}{kJ} \sum_{j=1}^J \hat{H}_{j|k} \,, %\left[ \hat{H}_{j|k} - \Psi_{k|i}\left(\bar{F}_{M,i-1} (\theta)\right) \right]
\end{aligned}
\end{equation}
where $\hat{H}_{j|k}$ is the $j$th estimate of the Hessian matrix, $H(\theta)$ at the sample data $\vec{z}_k$. The inner averaging loop compute an estimate of the Hessian matrix at a sample data $\vec{z}_k$ using $J$ values of perturbations, $\Delta_1, \Delta_2, \cdots, \Delta_J$. After executing \eqref{ModifiedFisher} $K$ times (for all samples $\vec{z}_k$), the $\bar{F}_{J,K}$ is used as an estimation of the Hessian matrix at $\theta$. The details of the Monte Carlo based Hessian estimation can be found in \shortcite{spall2005monte,spall2008improved}.

%%%%%%%%%%%%%%%%%%%%%%%%%%%%%%%%%%%%%%%%%%%%%%%%%%%%%%%%%%%%%
\section{PARALLEL SIMULTANEOUS PERTURBATION ALGORITHM: PSPO}  \label{sec:SPSApp}
%%%%%%%%%%%%%%%%%%%%%%%%%%%%%%%%%%%%%%%%%%%%%%%%%%%%%%%%%%%%%

In case of epidemiological model fitting, the objective function is obtained from \eqref{costLogReg}. In order to compute this log likelihood function, we need to run the SSA simulation given $\theta$ to compute the probability term \eqref{Pr_zj}. Since SSA simulation is a stochastic process, the log likelihood function (the objective function) is a noisy uncertain function. Although SPSA is an efficient optimization algorithm for high dimensional problems, it requires many iterations to converge, particularly in the case of high noise problems. Since all epidemic models are stochastic processes, the values of objective function are highly uncertain. In the case of highly uncertain problems, the computation of the gradient from noisy objective function require more runs to obtain an acceptable estimate. The idea, here, is using parallel computing of the gradient with different perturbation vectors. The multiple rounds of computation guarantees an error bound for each iteration given the noise variance presented in the objective function. The connection of the number of computation rounds and the error in estimated gradient is specified here.

The gradient vector at each iteration can be obtained by averaging over multiple evaluations of \eqref{est_g_1sided} for different values of perturbation. Let $\vec{g}^i$ represents the directional gradient at a point $\theta$ in $\Delta_i$ direction, and $\vec{g}$ is the gradient at this point. Then we know that:
$$ \vec{g}^i = \frac{\left< \vec{g}(\theta), \Delta_i \right>}{\| \Delta_i \|^2} \Delta_i \,.$$
Let $\hat{\vec{g}}^i$ is the estimated gradient using the one-sided gradient approximation \eqref{est_g_1sided} and perturbation vector $\Delta_i$ is a $\pm 1$ vector. Then,
$$ \hat{\vec{g}}^i \approx p \vec{g}^i = \left< \vec{g}(\theta), \Delta_i \right> \Delta_i \approx \left< \hat{\vec{g}}(\theta), \Delta_i \right> \Delta_i \,.$$
Then, we have
$$ \hat{\vec{g}}^T \Delta_i \approx \frac{\delta f_i}{c} \,; \ \ \delta f_i = L(\theta+c\Delta_i) - L(\theta)\,,$$
and then,
\begin{equation} \label{gradLinEq} 
\hat{\vec{g}}^T \Delta \approx \frac{1}{c} \begin{bmatrix} \delta f_1& \delta f_2& \cdots& \delta f_M \end{bmatrix}\,,
\end{equation}
where,
$$ \Delta = \begin{bmatrix} \Delta_1& \Delta_2& \cdots& \Delta_M \end{bmatrix}\,.$$
If $\{\Delta_i\}_{1\leq i\leq M}$ span $\mathbb{R}^p$, then $\Delta \Delta^T$ is invertible. So, the least square estimation of $\hat{\vec{g}}$ is given by:
\begin{equation} \label{MLS_grad}
\hat{\vec{g}} \approx \frac{1}{c} \left( \Delta \Delta^T \right)^{-1} \Delta \begin{bmatrix} \delta f_1 & \delta f_2 & \cdots & \delta f_M \end{bmatrix}^T\,.
\end{equation}

On the other hand, if $M<p$, then \eqref{gradLinEq} is an under-determined system which has non-unique solutions. Then, the solution of the minimum Euclidean norm, $\displaystyle \|\hat{\vec{g}}\|_{2}$, among all solutions is given by:
\begin{equation} \label{MN_grad}
\hat{\vec{g}} \approx \frac{1}{c} \Delta \left( \Delta^T \Delta \right)^{-1} \begin{bmatrix} \delta f_1 & \delta f_2 & \cdots & \delta f_M \end{bmatrix}^T\,.
\end{equation}

The PSP algorithm for estimating the gradient at a given point is given in Algorithm \ref{sp_pp}. In order to improve the efficiency of this algorithm, we can compute $L(\vec{\theta})$ once, out of the loop. Doing this, we need $M+1$ function evaluations for estimating the gradient.

\begin{algorithm}
Inputs: given point $\theta$, perturbation size $c$, and \# of parallel rounds $M$\;
Randomly sample $\Delta_0$ from $\{\pm 1\}$ binary distribution\;
Initialize computation round counter $i:=1$\;
\Repeat{$i \leq M$}{
	$\displaystyle \bar{i} = i \bmod p$\;
 	$\displaystyle \Delta_i := (I-2\vec{e}_{\bar{i}} \vec{e}_{\bar{i}}^T) \Delta_0$\;
	$\displaystyle \delta f_i(\theta) := L(\theta + c \Delta_i) - L(\theta) $\;
	$i := i+1$\;
}
$\displaystyle \Delta := \begin{bmatrix} \Delta_1& \Delta_2& \hdots & \Delta_M\end{bmatrix}$\;
\begin{algorithmic}
	\IF{$M\geq p$}
      		\STATE $\displaystyle \hat{\vec{g}}(\theta) :=\frac{1}{c} \left(\Delta \Delta^T \right)^{-1} \Delta \begin{bmatrix} \delta f_1 & \delta f_2 & \hdots & \delta f_M\end{bmatrix}^T$\;
	\ELSE
		\STATE $\displaystyle \hat{\vec{g}}(\theta) :=\frac{1}{c} \Delta \left(\Delta^T \Delta \right)^{-1} \begin{bmatrix} \delta f_1 & \delta f_2 & \hdots & \delta f_M\end{bmatrix}^T$\;
	\ENDIF
\end{algorithmic}

%Output: $\hat{\vec{g}}(\theta)$\;
 \vspace{-1mm}
\NoCaptionOfAlgo
 \caption{Algorithm 1: PSP gradient.} \label{sp_pp}
\end{algorithm}

One difference of the PSP algorithm compared with the conventional SPSA algorithm is on the selection of the perturbation vectors, $\Delta$. We suggest a more efficient strategy to choose independent $\Delta$ vectors. First we sample $\Delta_0$ from a Bernoulli $\pm 1$ distribution. In round $j$, $1 \leq j \leq p$, we switch the sign of the $j$th element in $\Delta_{0}$ to generate $\Delta_{j}$. The same procedure repeats for the next $p$ rounds of computations. In \shortcite{alaeddini2017pspo}, it is proved that this strategy generates a set of perturbation vectors that spans $\mathbb{R}^p$.

We need to provide the number of parallel rounds of computation, $M$, as an input for this algorithm. Let $y=L(\theta)$ and $y^i=L(\theta+c\Delta_i)$ represent the noisy measurements at $\theta$ and $\theta+c\Delta_i$ respectively. Assume these measurements can be expressed as normal random variables as $y \sim \mathcal{N}(\mu,\sigma^2)$ and $y^i \sim \mathcal{N}(\mu_i,\sigma^2)$. It is assumed that the noise variance is constant over the whole state space, $\theta$. 

\begin{lemma} \label{lem_expVal}
Given $y \sim \mathcal{N}(\mu,\sigma^2)$ and $y^i \sim \mathcal{N}(\mu_i,\sigma^2)$, if $M\geq p$, then the expected value of the gradient from Algorithm \ref{sp_pp} is equal to the gradient $\vec{g}$.
\end{lemma}

\begin{proof}
%Using \eqref{MLS_grad}, we have
%\begin{equation} \label{lem_eq1} 
%\mathbb{E} [\hat{\vec{g}}] = \frac{1}{c} \left(\Delta \Delta^T \right)^{-1} \Delta \begin{bmatrix}  \mathbb{E} \left[ \delta f_1\right] &  \mathbb{E} \left[ \delta f_2 \right]& \hdots &  \mathbb{E} \left[ \delta f_M\right]\end{bmatrix}^T \,.
%\end{equation}
%Based on our assumption, 
%$$ \mathbb{E} \left[ \delta f_i \right] = \mu_i - \mu \,.$$
%For small $c$, we have
%$$ \vec{g}^T \Delta_i = \frac{\mu_i - \mu}{c}\,. $$
%Then
%$$ \mathbb{E} \left[ \delta f_i \right] =c \vec{g}^T \Delta_i \,.$$
%After substituting this in \eqref{lem_eq1}, we have
%\begin{equation} \label{lem_eq2} 
%\mathbb{E} [\hat{\vec{g}}] = \frac{1}{c} \left(\Delta \Delta^T \right)^{-1} \Delta c \begin{bmatrix}  \Delta_1^T \\  \Delta_2^T \\ \vdots \\  \Delta_M^T \end{bmatrix} \vec{g} =  \left(\Delta \Delta^T \right)^{-1}  \left(\Delta \Delta^T \right) \vec{g} = \vec{g}\,.
%\end{equation}
Refer to \shortcite{alaeddini2017pspo}.
\end{proof}

\begin{theorem} \label{thrm_bound}
If the rounds of computation, $M$, satisfies
\begin{equation} \label{Nparallel}
M \geq \frac{\sigma^2 p}{c^2\epsilon^2}\,,
\end{equation}
then the error in estimated gradient from Algorithm \ref{sp_pp} is bounded by $ \mathbb{E} \left[ \| \hat{\vec{g}}- \vec{g}\| \right] \leq \epsilon \,.$
\end{theorem}

\begin{proof}
Refer to \shortcite{alaeddini2017pspo}.
%Since $\Delta_i$ is a $\pm 1$ vector, if we use one-sided gradient approximation \eqref{est_g_1sided}, the function difference $\delta f_i(\theta)$ can be expressed as:
%$$\delta f_i(\theta) \sim \mathcal{N}( \mu_i-\mu,2\sigma^2)\,.$$
%Using \eqref{MLS_grad} for estimating the gradient and \cref{lem_expVal}, the error in gradient, $(\hat{\vec{g}} - \vec{g})$, is a normal random variable given by:
%$$ (\hat{\vec{g}} - \vec{g}) \sim \mathcal{N} (0,\Sigma) \,,$$
%where
%$$ \Sigma= \frac{2\sigma^2}{c^2} \left(\Delta \Delta^T \right)^{-1} \Delta \left( \left(\Delta \Delta^T \right)^{-1} \Delta \right)^T = \frac{2\sigma^2}{c^2} \left(\Delta \Delta^T \right)^{-1} \,.$$
%After $M$ rounds of computation of $\hat{\vec{g}}$, Chebyshev inequality implies that
%$$ \Pr \left[ \| \hat{\vec{g}} - \vec{g} \| \geq \epsilon \right] \leq \frac{\trace(\Sigma)}{M \epsilon^2} = \frac{2 \sigma^2 \trace \left(\left(\Delta \Delta^T \right)^{-1}\right)}{c^2 M \epsilon^2} \,.$$
%Using the results of \cref{lem_trBound}, we have
%$$ \Pr \left[ \| \hat{\vec{g}} - \vec{g} \| \geq \epsilon \right] \leq \frac{\sigma^2 p}{2c^2 M \epsilon^2} \,.$$
%Therefore, if 
%$$M \geq \frac{\sigma^2 p}{c^2\epsilon^2}\,,$$
%then 
%$$ \Pr \left[ \| \hat{\vec{g}} - \vec{g} \| \geq \epsilon \right] \leq \frac{1}{2} \,.$$
%Therefore, $\displaystyle \mathbb{E}  \left[ \| \hat{\vec{g}} - \vec{g} \| \right] \leq \epsilon\,.$
\end{proof}

It is worthy to note that from the results of \cref{thrm_bound} one can show that if $M \rightarrow \infty$, then the error, $\epsilon \rightarrow 0$, which is equivalent to $ \mathbb{E} \left[ \| \hat{\vec{g}}- \vec{g}\| \right] =0$. Algorithm \ref{conjugateSP} presents a step-by-step summary of the PSPO algorithm, which is a conjugate gradient type algorithm. Details can be found in \shortcite{alaeddini2017pspo}.

\begin{algorithm} %[H]
%Inputs: clinical data, $\zeta_1, \zeta_2, \cdots, \zeta_n$\;
Initialize $i=0$, $k=0$ and initial guess $\theta_0$, Initialize update direction $\vec{d} =$ null\;
Run PSP to compute $\hat{\vec{g}}_0(\theta_0)$ with $M=1$\;
$\vec{r} := -\hat{\vec{g}}_0(\theta_0)$, \ \ $\vec{d} := \vec{r}$, \ \ $\vec{g}_{new} := \hat{\vec{g}}_0(\theta_0)$\;
\Repeat{stopping criteria satisfied}{
	Set \# of parallel rounds $M$ using error tolerance for this iteration\;
 	Run PSP to compute $\hat{\vec{g}}_k(\theta_k)$ using computed $M$\;
	$\displaystyle \vec{r} := -\hat{\vec{g}}_k(\theta_k)$, \ \ $\displaystyle \vec{g}_{old} := \vec{g}_{new}$, \ \ $\displaystyle \vec{g}_{new} :=\hat{\vec{g}}_k(\theta_k)$\;
	$\displaystyle \beta = \frac{\vec{g}_{new}^T(\vec{g}_{new}-\vec{g}_{old})}{\vec{g}_{old}^T \vec{g}_{old}}$\;
	
	$\displaystyle \tilde{\vec{d}} := \frac{\vec{d}}{\|\vec{d}\|} + \varepsilon \mathbbm{1}_{\vec{d}=0}(\vec{d})$\;
	$\theta^+_k := \theta_k + \tilde{\vec{d}}$, \ \ $\theta^-_k := \theta_k - \tilde{\vec{d}}$\;
 	Run PSP to compute $\hat{\vec{g}}(\theta^+_k)$, and $\hat{\vec{g}}(\theta^-_k)$, and compute $\displaystyle \delta G_k = \hat{\vec{g}}(\theta^+_k) - \hat{\vec{g}}(\theta^-_k)$\;
	$\displaystyle \hat{H}_k :=\frac{1}{2} \left[ \frac{\delta G_k}{2} \tilde{\vec{d}}^{-1} + \left( \frac{\delta G_k}{2} \tilde{\vec{d}}^{-1} \right)^T \right]$\;
	
	$\displaystyle \alpha :=  -\frac{\hat{\vec{g}}_k^T \vec{d}}{\vec{d}^T \hat{H}_k \vec{d}}$\;
	Update $\displaystyle \theta_{k+1} :=\theta_{k} +\alpha \vec{d}$, and update direction $\displaystyle \vec{d} :=\vec{r}+\beta \vec{d}$\;
	$i := i+1$\;
	\begin{algorithmic}
	\IF{$i =p$ or $\vec{r}^T\vec{d} \leq 0$}
      		\STATE $\displaystyle \vec{d} := \vec{r}$\;
		\STATE $\displaystyle i := 0$\;
	\ENDIF
	\end{algorithmic}
	$k := k+1$\;
}
%$\theta_{est} := \theta_k$ \;
%Outputs: $\theta_{est}$\;
 \vspace{0mm}
 \NoCaptionOfAlgo
 \caption{Algorithm 2: PSPO Algorithm.} \label{conjugateSP}
\end{algorithm}

Running the epidemiological simulations, particularly the computation of likelihood function for a complex model, is usually an expensive task. Therefore, it is more desirable to have a robust solution, that we do not require to re-run the whole simulation with any change in the model. On the other hand, we would like to have a region of all possible solutions. Then we would be able to sample from the solution region and only keep the parameters that often produce good runs for the given observation. In order to estimate the solution region, we use the Monte Carlo based approach, introduced in \shortcite{spall2005monte}, and described earlier in this paper. At this point, we have all required tools to find the best predictive model from observed clinical data. Algorithm \ref{SPSA_pp} presents a step-by-step summary of the proposed approach. 
\begin{algorithm} 
Inputs: clinical data, $\zeta_1, \zeta_2, \cdots, \zeta_n$\;
Initialize with given guess $\theta_0$\;
Run PSPO algorithm until the stopping criteria satisfied\;
$\theta_{est} := \theta_k$ obtained from PSPO\;
Compute $\displaystyle \bar{F}(\theta_{est})$ using \eqref{ModifiedFisher}\;
$\displaystyle \Sigma_{\theta_{est}} = \bar{F}^{-1}(\theta_{est})$\;
Sample from $\theta_{est}$ and $\Sigma_{\theta_{est}}$ to find a desirable set of parameters\;
\vspace{3mm}
\NoCaptionOfAlgo
\caption{Algorithm 3: Epidemic model calibration algorithm.} \label{SPSA_pp}
\end{algorithm}

%%%%%%%%%%%%%%%%%%%%%%%%%%%%%%%%%%%%%%%%%%%%%%%%%%%%%
\section{SIMULATION} \label{sec:sims}
%%%%%%%%%%%%%%%%%%%%%%%%%%%%%%%%%%%%%%%%%%%%%%%%%%%%%

To demonstrate the main contribution of this paper, we implemented the methodology on Hagelloch data set. This data set concerns a measles outbreak in a small town in Germany in 1861, which contains 188 infected individuals. This data set is very popular in literature because of its completeness and depth of data. \cref{Measles1861} gives the observed clinical data. The data are obtained from \shortcite{meyer2014spatio}.
%Simulations run on a supercomputer running Windows HPC Server.
%
%%% ==============================      example 1     ===============================%%
%\subsection{Example One: Minimization of Quadratic Function}
%
%The basic concept of the algorithm can be seen more clearly on a simple toy problem. The objective function, here, is
%$$ f(x,y,z) = -(x-1)^2 - (y-1)^2 -(z-1)^2 + w\,,$$
%where, $w$ is a Gaussian noise $w \sim \mathcal{N}(0,0.02^2)$. The result of the PSPO optimization algorithm from different initial conditions is plotted in \cref{trajs2ndReg}. The optimal point of this function is known and we can easily test our algorithm.
%\begin{figure}[!h]
%        \centering
%        \includegraphics[width=.5\linewidth]{trajs2ndReg.pdf}
%        \caption{Finding the minimum point of a three-dimensional quadratic function using noisy measurement.} \label{trajs2ndReg}
%\end{figure}
%
%%% ==============================      example 3     ===============================%%
%\subsection{Example Two: Predictive Model for Hagelloch Measles Data}
%We consider a data set from a measles outbreak in a small town in Germany in 1861    

\begin{figure}[!htb]
        \centering
        \includegraphics[width=.45\linewidth]{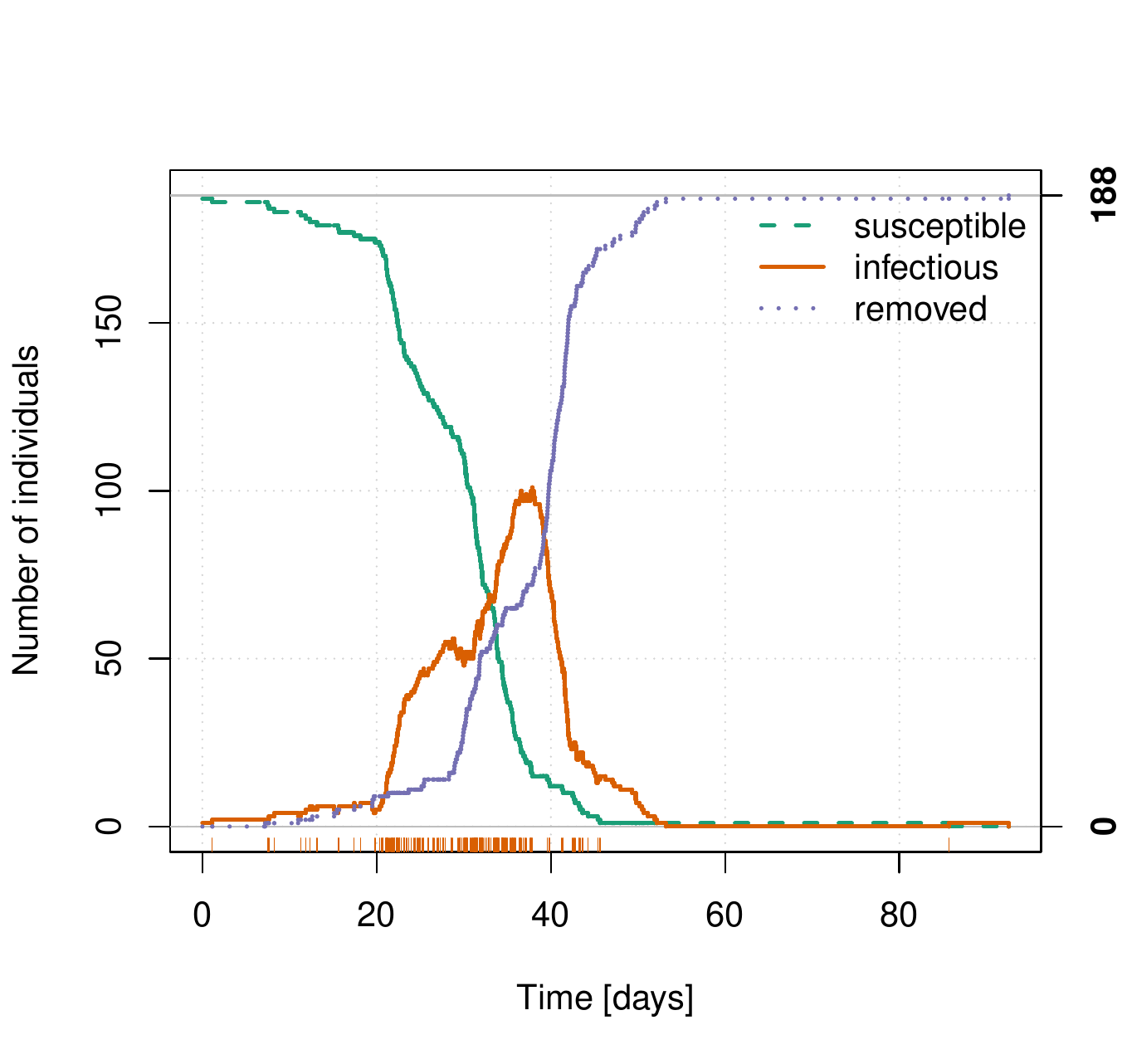}
	\vspace{-5mm}
        \caption{Measles outbreak clinical data.} \label{Measles1861}
\end{figure}

Some types of diseases with permanent immunity, such as measles, mumps and rubella, can be described as Susceptible-Infected-Recovered (SIR) model. In this model each individual can only exist in one of the discrete states such as susceptible (S), infected (I) or permanently recovered (R). We have two transitions in this case. An infected person can infect others with an infection rate, $\beta$, and is cured with curing rate, $\delta$. Therefore, the parameters of our model are $ \theta = (\beta, \delta)$. Our objective, here, is finding a good model for the given epidemic data, i.e. infection rate and curing rate, to investigate the properties of the disease spread. These properties will allow the researchers to learn about the diseases, and thereby enabling them to test competing theories about transmission of disease and to devise better containment strategies.

%In order to evaluate the algorithm for this more complex epidemiological model, we run the simulation for each point $(\beta,\delta)$ 1000 times and average the value of log likelihood value as the objective function value. The averaged objective function is given in \cref{costfunc_1000sims}. Note that this averaged objective function are not used in the optimization algorithm, and it is only used for the evaluation of the algorithm.
%\begin{figure}[!h]
%        \centering
%        \includegraphics[width=.5\linewidth]{costfunc_1000sims.pdf}
%        \caption{Noiseless objective function averaged of 1000 simulations for each point.} \label{costfunc_1000sims}
%\end{figure}

We compare the efficiency of the PSPO algorithm with the conventional SPSA algorithm. To compare these two algorithms, we run the optimization algorithms 100 times for each one. \cref{hist19} shows the number of iterations required to converge. As it can be seen in this figure, PSPO algorithm in average needs fewer iterations compared with the conventional SPSA algorithm. The maximum number of iterations for both algorithms set to 30. Note that the number of iterations until convergence set to 30 if the algorithm does not converge after 30 iterations. In order to investigate the performance of the PSPO algorithm as the number of parallel rounds increases, we run the PSPO for this example with different number of parallel computation rounds. \cref{fig:conv_ver_M} shows the number of iterations for convergence as $M$ grows.
%\begin{figure}[!h]
%        \centering
%        \includegraphics[width=.5\linewidth]{PSPOvsSPSA_100Iter_measles.pdf}
%%	\vspace{-10mm}
%        \caption{Comparison of the number of iterations for convergence for conventional SPSA and PSPO} \label{hist19}
%\end{figure}
\begin{figure}[!h]
\centering
\begin{subfigure}{.5\textwidth}
  \centering
  \includegraphics[width=\linewidth]{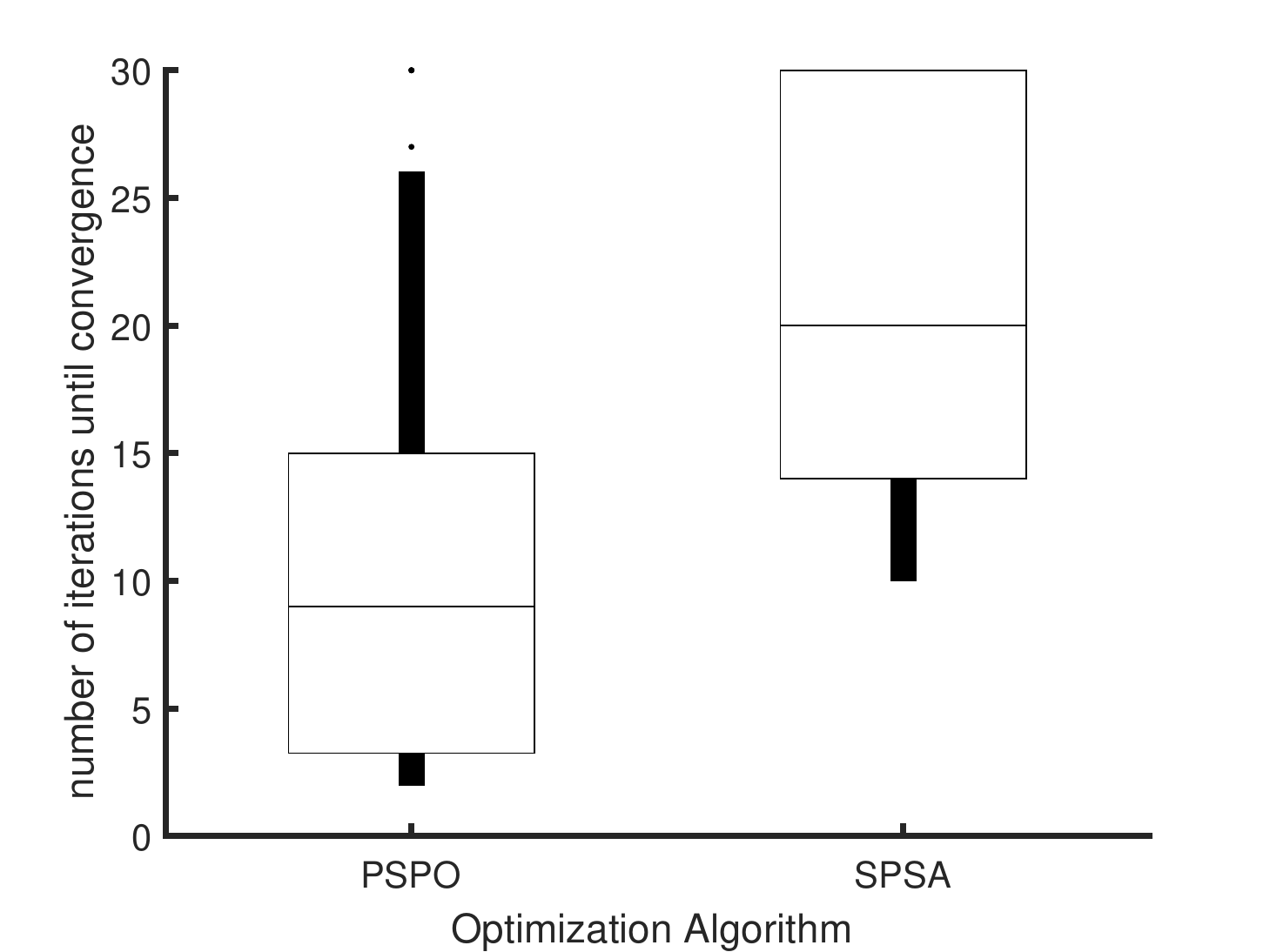}
  \caption{} \label{hist19}
\end{subfigure}%
\begin{subfigure}{.5\textwidth}
  \centering
  \includegraphics[width=\linewidth]{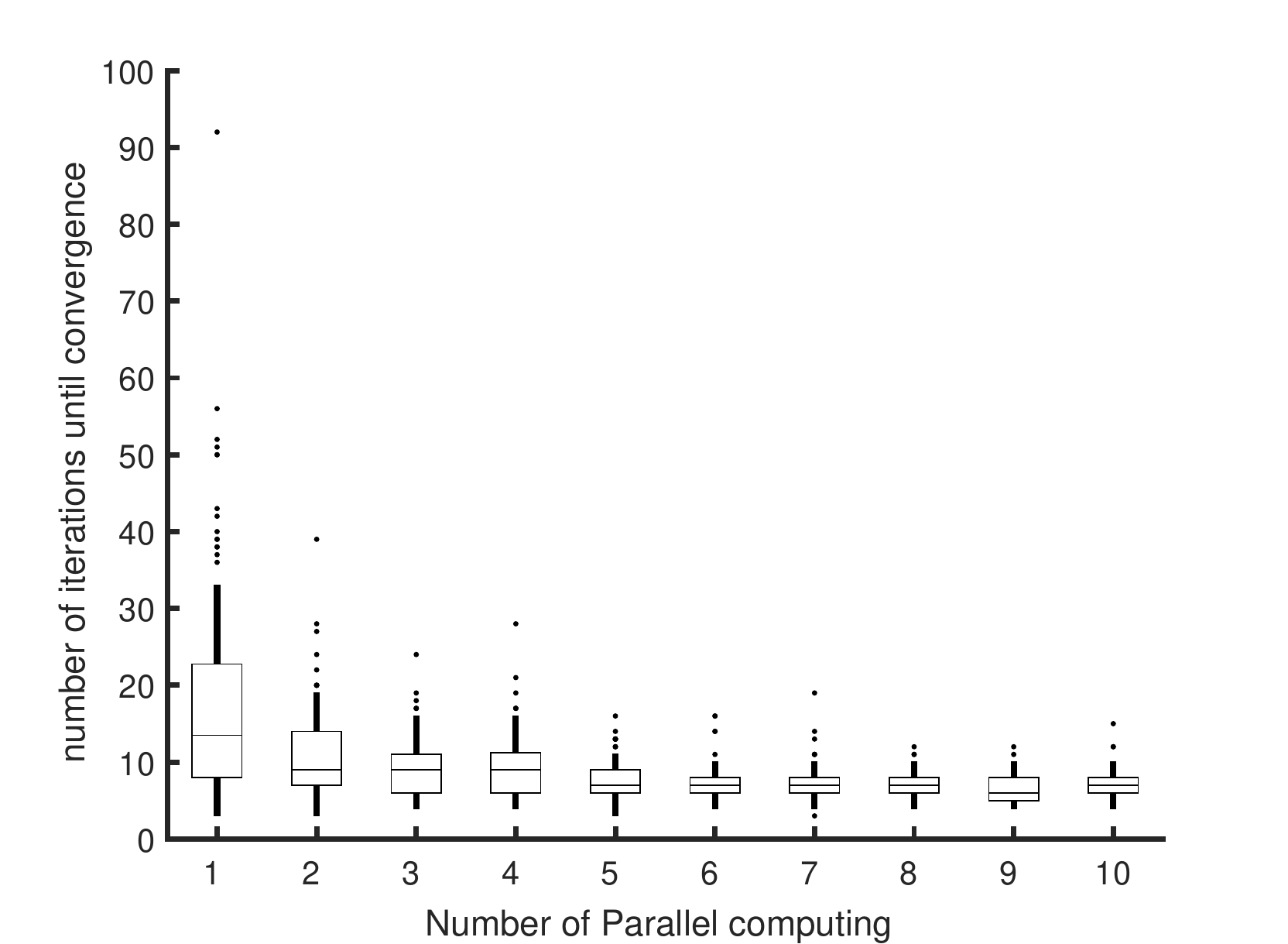}
\caption{}
\label{fig:conv_ver_M}
\end{subfigure}
\caption{(a) Comparison of the number of iterations for convergence for conventional SPSA and PSPO. (b) Number of iterations vs. number of parallel computing rounds, $M$, in PSPO.}
\label{fig:test}
\end{figure}

The comparison between SPSA and PSPO algorithm demonstrates the superiority of the PSPO algorithm in the number of iterations for convergence. Note that, comparing with SPSA, the PSPO algorithm with multiple parallel computations requires additional processing time per iteration. In order to prevent unnecessary long processing time per iteration, which results in long total processing time, the number of parallel computations need to be carefully computed using the level of the noise presented in the problem. In some cases when we are dealing with a high SNR (signal to noise ratio) problem, a single round of computation might be enough for convergence, but in some other applications, like most of the epidemiological applications, the data are highly noisy, and SPSA requires too many iterations to converge.

%%%%%%%%%%%%%%%%%%%%%%%%%%%%%%%%%%%%%%%%%%%%%%%%%%%%%%%%%%%%%%
\section{CONCLUSIONS} \label{sec:conclusion}
%%%%%%%%%%%%%%%%%%%%%%%%%%%%%%%%%%%%%%%%%%%%%%%%%%%%%%%%%%%%%%

In many stochastic optimization problems that noise presents, obtaining an analytical solution is hardly possible. In this paper, we described an algorithm for optimal parameter estimation in discretely observed stochastic epidemic model. We presented an efficient algorithm which explores the parameters space to find a set of parameters which gives highest conditional probability of the given noisy measurements. The log likelihood function, which was used to evaluate the cost corresponding to different choices of parameters, was estimated by running the Gillespie stochastic simulation algorithm and then using the beta binomial probabilistic model. We developed a parallel stochastic perturbation optimization algorithm and found the minimum number of parallel rounds of computation in order to have a bounded error in the estimated gradient. Furthermore, using a Monte Carlo re-sampling method, we computed the fisher information matrix, which gives the confidence region of the obtained solution. 

\section*{ACKNOWLEDGMENTS}
The authors thank Bill and Melinda Gates for their active support and their sponsorship through the Global Good Fund. The authors would also like to thank Philip Eckhoff for his helpful comments on the paper.

% Please don't exchange the bibliographystyle style
\bibliographystyle{wsc}
% AUTHOR: Include your bib file here
\bibliography{citations}

\section*{AUTHOR BIOGRAPHIES}

\noindent {\bf ATIYE ALAEDDINI} is a Postdoctoral Research Scientist at the Institute for Disease Modeling, Applied Math team. She holds a Ph.D. in Aerospace engineering from University of Washington and a M.Sc. in Computer engineering from University of California, Santa Cruz. Her research interests include stochastic optimization, control theory, online optimization algorithms, networked systems and their applications in epidemiology. As a member of IDM’s research team, she is developing a statistical model and optimization algorithm for model calibration and parameter space exploration. Her e-mail address is \email{aalaeddini@idmod.org}.\\

\noindent {\bf DANIEL KLEIN} is Sr. Research Manager at the Institute of Disease Modeling. He holds both Ph.D. and M.Sc. in Aerospace engineering from University of Washington. As the head of the Applied Math team at IDM’s research team, his time is split between developing the HIV model and building research tools. His HIV work focuses on building the contact network using ideas from feedback control to guide the dynamical relationship formation process. He is exploring structural assumptions in HIV modeling, HIV model calibration, the impact of interventions, and parameter sensitivity. His email address is \email{dklein@idmod.org}. \\

\newpage

\end{document}